\documentclass[10pt]{article}

\usepackage[margin=0.78in]{geometry}
\usepackage[T1]{fontenc}
\usepackage{amsmath,amssymb,amsfonts,mathtools}
\usepackage{amsthm}
\usepackage{graphicx}
\usepackage{booktabs}
\usepackage{placeins}
\usepackage{microtype}
\usepackage{enumitem}
\usepackage[dvipsnames]{xcolor}
\usepackage[colorlinks=true,allcolors=MidnightBlue]{hyperref}
\usepackage[round,authoryear]{natbib}

\setlist[itemize]{leftmargin=1.45em,itemsep=1pt,topsep=2pt}
\newtheorem{theorem}{Theorem}
\newtheorem{proposition}[theorem]{Proposition}
\newtheorem{lemma}[theorem]{Lemma}

\theoremstyle{remark}

\title{Physical Extinction and Long-Run Pricing\\
under Time-Varying Beliefs}
\author{Sourav Majumdar\\
\texttt{souravm@iitk.ac.in}\\
Department of Management Sciences, Indian Institute of Technology Kanpur}
\date{}

\begin{document}
\maketitle

\begin{abstract}
An investor may be optimistic about aggregate endowment growth at some times
and pessimistic at others. The weight placed on her forecast in bond valuation
can therefore vary across maturities. We study whether this maturity dependence
disappears at the long end of the yield curve. In a two-investor Arrow--Debreu economy,
physical extinction follows when \emph{total disagreement} grows without bound. We find that bond
valuation depends on net disagreement. If \emph{net disagreement} has no limit, neither the long-forward measure nor
the long bond exists. Bond yields can converge along the same belief path. The
long yield alone therefore cannot reveal whether valuation weights and bond
returns converge across maturities. We also obtain
finite-maturity error bounds for bond prices and forward densities.
\end{abstract}



\section{Introduction}\label{sec:introduction}

The price of a long-dated claim depends on how investors value future
consumption. Different growth forecasts imply different assessments of future
endowment paths. Through risk sharing, forecast errors alter consumption and
marginal utility. The influence of each forecast on equilibrium prices changes
with the resulting allocation.

Physical survival refers to an investor's share of aggregate consumption.
Valuation relevance concerns the effect of her growth forecast on current
prices. An additional unit of consumption has high marginal utility when
aggregate endowment is low. An investor who considers these states more likely
may therefore affect prices after her consumption share vanishes.

\citet{kogan2006price} and \citet{cvitanic2011price} establish this distinction.
We allow belief disagreement to change over time. A bond maturing at $T$
reflects the forecast errors accumulated before $T$. If their direction
changes, later errors can offset earlier ones and alter the relative weights
placed on the two forecasts across maturities. We ask whether these weights
converge and define a stable valuation law for increasingly distant payments.

A zero-coupon bond pays one unit of consumption at time $T$. Its yield is the
annualized discount rate implied by its price, and the long yield is the limit
of this rate as $T$ increases. The maturity-$T$ bond also defines valuation
weights for states at $T$. For a particular state, this weight is the price of
a claim that pays one unit of consumption when that state occurs and zero
otherwise, expressed as a share of the bond price. These weights define the
maturity-$T$ forward measure.

Suppose aggregate endowment falls before time $t$. Consider all states at
$T>t$ in which this decline has occurred. Their combined weight is the weight
that the maturity-$T$ forward measure assigns to the decline. This weight may
change with the payment time $T$. We ask whether it converges as $T$ moves
farther into the future. If the weights assigned to every event over each
finite time interval converge, their limits define the long-forward measure.
In our economy, this law shows how the two investors' forecasts enter long-run
pricing.

Bond prices also determine the returns on investments in bonds of different
maturities. The time-$t$ price of a maturity-$T$ bond relative to its time-zero
price is the value at $t$ of one unit invested in the bond initially. If these
return processes converge as $T$ increases, their limit is the long bond. When
the long-forward measure and long bond exist, they identify the
Hansen--Scheinkman decomposition of the stochastic discount factor
\citep{hansen2009long,qin2017long}. Distant bonds may share the same limiting
average discount rate. Their
valuation weights and returns can still depend on the payment time. A long
yield may therefore exist in the absence of a long-forward measure and long
bond.

We study these issues in a two-investor Arrow--Debreu economy. The investors
differ only in their forecasts of aggregate endowment growth and share CRRA
preferences and a common discount rate. The rational investor uses the
probability law that generates aggregate endowment. At time $t$, the
dogmatist's expected-growth forecast differs from the rational investor's
forecast by $\sigma^2\eta_t$, where $\sigma$ is the volatility of the endowment process and the path $\eta_t$ is deterministic denoting the evolution of the belief of the investor.

We describe the history of forecast disagreement using
\begin{equation}\label{eq:intro-statistics}
 A_t:=\int_0^t\eta_s^2\,\mathrm{d} s,
 \qquad
 H_t:=\int_0^t\eta_s\,\mathrm{d} s.
\end{equation}
We call $A_t$ total disagreement and $H_t$ net disagreement.
Total disagreement accumulates squared forecast errors. Net disagreement carries their signs, so errors in opposite directions can offset one another.
The distinction between these quantities is crucial to our analysis of physical
extinction and long-run pricing. For $\gamma=1$, we derive exact
finite-maturity formulas and use them in the numerical illustrations.

\subsection{Relation to the literature}

\citet{sandroni2000markets} provide the benchmark for our extinction result.
Their investors share a discount factor and can trade a complete set of
contingent claims. Accurate forecasters retain consumption of the endowment over time.
Repeated forecast errors cause an investor's consumption share to vanish.
\citet{blume2006smart} show that this prediction can change once investors
cannot trade a complete set of claims. The securities available to trade then affect
which beliefs survive. \citet{yan2008natural} proves that an investor survives if and only if her survival index is
the lowest in the economy. The index includes the squared forecast error and
the time discount rate. The CRRA coefficient enters through aggregate growth.
With common preferences, only investors with the most accurate forecasts can
survive. Yan holds each investor's forecast error constant over time. In our model, the
investors share CRRA preferences and a common discount rate, and the forecast
error changes over time. Proposition \ref{prop:equilibrium} states the
resulting extinction criterion in terms of total disagreement $A_t$. The
equilibrium allocation then presents our analysis of
long-run pricing. \citet{borovicka2020survival} shows that this conclusion depends on
the form of common preferences. Under homothetic recursive preferences, an
investor with a less accurate forecast can dominate when risk aversion exceeds
the inverse of the intertemporal elasticity of substitution.

\citet{kogan2006price} show how price influence can outlast physical survival.
An investor whose wealth share approaches zero may still affect state and stock
prices when her consumption is concentrated in rare states with low aggregate
consumption. In these states, an additional unit of consumption has high
marginal utility. An investor's assessment of the likelihood of these states
can therefore affect current prices. Our question concerns the convergence of
the valuation law across increasingly distant bond maturities.

\citet{cvitanic2011price} use the $T$-forward measure associated with the
maturity-$T$ bond to show that an investor can affect bond risk after her
consumption share converges to zero under the physical law. Their result is
maturity-specific. Theorem \ref{thm:main} derives a pathwise condition for
forward measures and normalized bond prices to converge as the payment time
increases. In our economy, total disagreement $A_T$ determines whether the
dogmatist's consumption share converges to zero. For a bond maturing at $T$,
her relative contribution to valuation depends on net disagreement $H_T$. A
change in the sign of the forecast error can therefore alter her contribution
across bond maturities even when her consumption share converges to zero.

\citet{malamud2008long} show that a distant bond is valued asymptotically by
the investor whose expected marginal utility assigns it the highest value.
This investor has the lowest yield in the corresponding single-investor
economy. \citet{cvitanic2012financial} obtain a related long-yield result from
investors' saving motives. They explain how physical survival affects
which investor influences different parts of the forward curve. When the
difference between the two growth forecasts is constant, the investor with the
lower single-investor long yield determines the equilibrium long yield.

For every
$t<T$, Theorem \ref{thm:main}(i) approximates the bond price and forward
density by the rational and dogmatist benchmark components. It also bounds the
expected absolute error in the forward-density approximation. Net disagreement $H_{t,T}$ determines the
relative contribution of each component. This finite-maturity result
complements the asymptotic bond-yield comparisons in
\citet{malamud2008long} and \citet{cvitanic2012financial}.

\citet{hansen2009long} derive a factorization of a Markovian stochastic
discount factor from a positive eigenfunction of the pricing operator. The
construction produces a martingale component that defines a change of
probability measure. Positive eigenfunctions need not be unique.
\citet{qin2016positive} prove that the recurrent eigenfunction, when it exists,
is unique. In our space-time representation, the time coordinate increases
deterministically, and therefore cannot satisfy this
recurrence condition. We use the maturity-limit approach of
\citet{qin2017long} to identify the martingale component. They characterize
the long-forward measure through the limit of $T$-forward measures and the
long bond through the limit of normalized prices of increasingly long-maturity
bonds. \citet{qin2018hjm} give sufficient conditions for these limits to exist
in Heath--Jarrow--Morton models.

Theorem \ref{thm:main}(ii) applies these maturity limits to time-varying
disagreement. Physical extinction alone does not ensure convergence of
forward measures or normalized bond prices. Suppose total disagreement
diverges and net disagreement has no limit. The dogmatist's consumption share
converges to zero, yet her valuation weight does not converge across
maturities. The forward measures and normalized bond prices consequently have
no limit. If net disagreement approaches a finite value or grows without bound
in one direction, the theorem identifies the martingale component in the
resulting Hansen--Scheinkman factorization. Part (iii) gives the long yield
under the additional condition that $H_T/T$ converges. Net disagreement may
lack a limit even when this average converges to zero. Along such paths, the
long yield equals $r^R$, and no long-forward measure or long bond exists.

\section{The economy and physical extinction}\label{sec:model}

We work on a probability space
$(\Omega,\mathcal F,(\mathcal F_t)_{t\geq0},P)$. The filtration is generated by
the Brownian motion $(B_t)_{t\geq0}$. Aggregate endowment follows
\begin{equation}\label{eq:endowment}
 \frac{\,\mathrm{d} D_t}{D_t}=\mu\,\mathrm{d} t+\sigma\,\mathrm{d} B_t,
 \qquad D_0>0.
\end{equation}
The constant $\mu\in\mathbb R$ is its expected growth rate, and $\sigma>0$ is
its volatility.

The rational investor is denoted by $R$ and the dogmatist by $N$. Investor
$i\in\{R,N\}$ evaluates future endowment paths under the probability measure
$Q^i$ and consumes $c_{i,t}$ at time $t$. Both investors have time-additive
CRRA preferences with discount rate $\rho>0$ and relative risk aversion
$\gamma>0$. Investor $i$'s expected utility is
\begin{equation}\label{eq:preferences}
 \mathbb{E}^{Q^i}\left[\int_0^\infty e^{-\rho t}u(c_{i,t})\,\mathrm{d} t\right],
 \qquad
 u(c)=\begin{cases}
 c^{1-\gamma}/(1-\gamma),&\gamma\neq1,\\
 \log c,&\gamma=1.
 \end{cases}
\end{equation}

The rational investor's beliefs coincide with the physical probability
measure, so $Q^R=P$. The dogmatist's expected growth forecast differs from the
rational investor's forecast by $\sigma^2\eta_t$, where $\eta$ is a
deterministic function satisfying
\begin{equation}\label{eq:eta-integrability}
 \int_0^T\eta_s^2\,\mathrm{d} s<\infty,
 \quad \forall T>0.
\end{equation}
For $0\leq t\leq T$, define
\begin{equation}\label{eq:disagreement-increments}
 A_{t,T}:=\int_t^T\eta_s^2\,\mathrm{d} s,
 \qquad
 H_{t,T}:=\int_t^T\eta_s\,\mathrm{d} s,
 \qquad A_t:=A_{0,t},\quad H_t:=H_{0,t}.
\end{equation}
We call $A_{t,T}$ total disagreement and $H_{t,T}$ net disagreement over
$[t,T]$. The same names apply to $A_t$ and $H_t$ over $[0,t]$.
The dogmatist's probability measure is defined on each $\mathcal F_t$ by the
likelihood ratio
\begin{equation}\label{eq:likelihood}
 \xi_t:=\left.\frac{\,\mathrm{d} Q^N}{\,\mathrm{d} P}\right|_{\mathcal F_t}
 =\exp\left\{\sigma\int_0^t\eta_s\,\mathrm{d} B_s
 -\frac{\sigma^2}{2}A_t\right\}.
\end{equation}
Since $\eta_t$ is deterministic, condition \eqref{eq:eta-integrability} makes
$\xi_t$ a martingale with mean one on every finite horizon. The measures $Q^N$
and $P$ are therefore equivalent on each $\mathcal F_t$.
The process $(B_t^N)_{t\geq0}$ defined by
$B_t^N:=B_t-\sigma H_t$ is a Brownian motion under $Q^N$ by Girsanov's
theorem. Hence the dogmatist investor perceives the expected endowment growth rate to be
$\mu+\sigma^2\eta_t$.

We consider an Arrow--Debreu economy, as in
\citet[Section 2]{cvitanic2012financial}. Each investor chooses a consumption of the endowment, $c_i=(c_{i,t})_{t\geq0}$, to maximize \eqref{eq:preferences} subject to
\begin{equation}\label{eq:consumption-value}
 \mathbb{E}^P\left[\int_0^\infty\pi_tc_{i,t}\,\mathrm{d} t\right]\leq W_i,
\end{equation}
where $\pi_t$ is a strictly positive state-price density and $W_i>0$ is investor
$i$'s initial wealth. We impose a feasibility constraint requiring $c_{R,t}+c_{N,t}=D_t$ at each time.

The parameter
\[
 w_0:=\frac{c_{N,0}}{c_{R,0}}>0
\]
determines the initial division of consumption between the two investors.
Proposition \ref{prop:equilibrium} gives the equilibrium allocation for each
value of $w_0$.

We impose the following conditions to ensure finite equilibrium wealth. For
$\gamma\neq1$, assume
\begin{equation}\label{eq:finite-power}
 \int_0^\infty e^{-\rho t}\mathbb{E}^P\left[
 (1+\xi_t)\left(\frac{D_t}{D_0}\right)^{1-\gamma}
 \right]\,\mathrm{d} t<\infty.
\end{equation}
For $\gamma=1$, assume
\begin{equation}\label{eq:finite-log}
 \int_0^\infty e^{-\rho t}\mathbb{E}^P\left[(1+\xi_t)
 \left\{1+\left|\log\frac{D_t}{D_0}\right|+|\log\xi_t|\right\}
 \right]\,\mathrm{d} t<\infty.
\end{equation}

\begin{proposition}[Arrow--Debreu equilibrium and physical extinction]
\label{prop:equilibrium}
Fix an initial consumption ratio $w_0>0$. Define
\begin{equation}\label{eq:sharing}
 \frac{c_{N,t}}{c_{R,t}}=w_t:=w_0\xi_t^{1/\gamma},
 \qquad
 c_{R,t}=\frac{D_t}{1+w_t},\qquad
 c_{N,t}=\frac{w_tD_t}{1+w_t}.
\end{equation}
With $\pi_0=1$, the equilibrium state-price density is
\begin{equation}\label{eq:sdf}
 \pi_t=e^{-\rho t}\left(\frac{D_t}{D_0}\right)^{-\gamma}
 \left(\frac{1+w_t}{1+w_0}\right)^\gamma.
\end{equation}
Under \eqref{eq:finite-power} or \eqref{eq:finite-log}, as applicable, the
wealth levels
\begin{equation}\label{eq:equilibrium-wealth}
 W_i(w_0):=\mathbb{E}^P\left[\int_0^\infty\pi_tc_{i,t}\,\mathrm{d} t\right],
 \qquad i\in\{R,N\},
\end{equation}
are finite and strictly positive. The allocation in \eqref{eq:sharing} and the
state-price density in \eqref{eq:sdf} form an Arrow--Debreu equilibrium for
initial wealth levels $W_i=W_i(w_0)$.

The dogmatist's consumption share converges to zero $P$-almost surely if and
only if $A_t\to\infty$. In that case,
\begin{equation}\label{eq:extinction-rate}
 \frac{\log w_t}{A_t}\longrightarrow-\frac{\sigma^2}{2\gamma}
 \qquad P\text{-a.s.},
\end{equation}
\end{proposition}

\begin{proof}
Under $P$, the rational investor's and the dogmatist's objective functions are
\begin{align*}
 U_R(c_R)&=\mathbb{E}^P\left[\int_0^\infty
 e^{-\rho t}u(c_{R,t})\,\mathrm{d} t\right],\\
 U_N(c_N)&=\mathbb{E}^P\left[\int_0^\infty
 e^{-\rho t}\xi_tu(c_{N,t})\,\mathrm{d} t\right].
\end{align*}
The likelihood ratio $\xi_t$ is present in the second objective function since the
dogmatist evaluates consumption under $Q^N$.
For the allocation in \eqref{eq:sharing}, set
$y_R=c_{R,0}^{-\gamma}$ and $y_N=c_{N,0}^{-\gamma}$. Since
$c_{N,0}=w_0c_{R,0}$, these positive constants satisfy
$y_N=w_0^{-\gamma}y_R$. Define
\[
 \pi_t:=y_R^{-1}e^{-\rho t}c_{R,t}^{-\gamma}.
\]
Then $\pi_0=1$, and substitution from \eqref{eq:sharing} gives
\eqref{eq:sdf}. The identity $w_t^\gamma=w_0^\gamma\xi_t$ then gives
\begin{equation}\label{eq:budget-first-order}
 e^{-\rho t}u'(c_{R,t})=y_R\pi_t,
 \qquad
 e^{-\rho t}\xi_tu'(c_{N,t})=y_N\pi_t.
\end{equation}
These are the marginal-utility conditions for the two individual consumption
problems.

We next verify finiteness. For $\gamma\neq1$, the identity
$w_t^\gamma=w_0^\gamma\xi_t$ and the bound
$(1+x)^\gamma\leq C(1+x^\gamma)$ give
\[
 0<\pi_tc_{i,t}\leq\pi_tD_t
 \leq C e^{-\rho t}(1+\xi_t)
 \left(\frac{D_t}{D_0}\right)^{1-\gamma},
\]
and
\begin{align*}
 e^{-\rho t}\left(c_{R,t}^{1-\gamma}
 +w_0^\gamma\xi_tc_{N,t}^{1-\gamma}\right)
 &=e^{-\rho t}D_t^{1-\gamma}(1+w_t)^\gamma\\
 &\leq C e^{-\rho t}(1+\xi_t)
 \left(\frac{D_t}{D_0}\right)^{1-\gamma}.
\end{align*}
For $\gamma=1$, substitution of $w_t=w_0\xi_t$ gives
\[
 \pi_tD_t+e^{-\rho t}
 \left(|\log c_{R,t}|+\xi_t|\log c_{N,t}|\right)
 \leq C e^{-\rho t}(1+\xi_t)
 \left(1+\left|\log\frac{D_t}{D_0}\right|+|\log\xi_t|\right).
\]
The applicable integrability condition therefore makes both utilities and the
wealth levels in \eqref{eq:equilibrium-wealth} finite. The wealth levels are
strictly positive since $\pi_t$ and $c_{i,t}$ are strictly positive.

We next verify individual optimality. Let $\widetilde c_R$ and
$\widetilde c_N$ be consumption plans that satisfy the respective budget constraints and
have well-defined utility. Concavity of $U$ and
\eqref{eq:budget-first-order} give
\begin{align*}
 U_R(\widetilde c_R)-U_R(c_R)
 &\leq y_R\mathbb{E}^P\left[\int_0^\infty\pi_t
 (\widetilde c_{R,t}-c_{R,t})\,\mathrm{d} t\right]\leq0,\\
 U_N(\widetilde c_N)-U_N(c_N)
 &\leq y_N\mathbb{E}^P\left[\int_0^\infty\pi_t
 (\widetilde c_{N,t}-c_{N,t})\,\mathrm{d} t\right]\leq0.
\end{align*}
Thus each investor optimally chooses the allocation in \eqref{eq:sharing}.
Equation \eqref{eq:sharing} also gives
$c_{R,t}+c_{N,t}=D_t$ at every time and state.

For the extinction claim, define the continuous local martingale
\[
 \mathcal M_t:=\int_0^t\eta_s\,\mathrm{d} B_s,
 \qquad \langle\mathcal M\rangle_t=A_t.
\]
The likelihood ratio and \eqref{eq:sharing} give
\[
 \log w_t=\log w_0+\frac{\sigma}{\gamma}
 \mathcal M_t-\frac{\sigma^2}{2\gamma}A_t.
\]
When $A_t\to\infty$, the Dambis--Dubins--Schwarz theorem gives a Brownian motion
$(\widetilde B_s)_{s\geq0}$ such that
$\mathcal M_t=\widetilde B_{A_t}$. The Brownian strong law gives
$\widetilde B_s/s\to0$ almost surely. Hence
$\mathcal M_t/A_t\to0$ almost surely. The identity
\[
 \frac{\log w_t}{A_t}
 =\frac{\log w_0}{A_t}+\frac{\sigma}{\gamma}\frac{\mathcal M_t}{A_t}
 -\frac{\sigma^2}{2\gamma}
\]
proves \eqref{eq:extinction-rate}. It also gives
$\log w_t\to-\infty$, so $w_t/(1+w_t)\to0$ almost surely. If $A_t$ has a finite
limit, then $\mathcal M_t$ converges almost surely to a finite random variable.
The expression for $\log w_t$ then has a finite limit, and the limiting
consumption share lies strictly between zero and one. This proves the converse.
\end{proof}

\section{Hansen--Scheinkman factorizations and long-run pricing}\label{sec:pricing}

Proposition \ref{prop:equilibrium} shows that the dogmatist's share of aggregate
consumption vanishes when total disagreement, $A_t$, becomes large. Her belief
may still influence distant prices. We study this influence through the
Hansen--Scheinkman factorization \citep{hansen2009long,qin2017long}.

To obtain the benchmark quantities, consider each investor in isolation, so
consumption equals $D_t$. The rational investor's marginal-utility state-price
density is $e^{-\rho t}(D_t/D_0)^{-\gamma}$. Expressing the dogmatist's density
under $P$ multiplies this term by $\xi_t$. Equations \eqref{eq:endowment} and
\eqref{eq:likelihood} give
\begin{align*}
 e^{-\rho t}(D_t/D_0)^{-\gamma}
 &=e^{-r^Rt}M_t^R,\\
 e^{-\rho t}\xi_t(D_t/D_0)^{-\gamma}
 &=e^{-r^Rt-\gamma\sigma^2H_t}M_t^N.
\end{align*}
The rational-investor benchmark yield is
\begin{equation}\label{eq:rR}
 r^R:=\rho+\gamma\mu-\frac12\gamma(\gamma+1)\sigma^2.
\end{equation}
The corresponding rational and dogmatist martingale components are
\begin{align}
 M_t^R&:=\exp\left\{-\gamma\sigma B_t
 -\frac12\gamma^2\sigma^2t\right\},
 \label{eq:MR}\\
 M_t^N&:=\exp\left\{\sigma\int_0^t(\eta_s-\gamma)\,\mathrm{d} B_s
 -\frac{\sigma^2}{2}\int_0^t(\eta_s-\gamma)^2\,\mathrm{d} s\right\}.
 \label{eq:MN}
\end{align}
Since $M^R$ has mean one, the rational benchmark price of a bond maturing at $T$ is
$e^{-r^RT}$, which gives the constant yield $r^R$. The two martingale
components define the benchmark forward laws $P^{L,R}$ and $P^{L,N}$. Identifying a
long-forward law with either probability law requires
additional restrictions \citep{borovicka2016misspecified}.

Combining \eqref{eq:endowment} and \eqref{eq:likelihood} with
Proposition \ref{prop:equilibrium} gives
\begin{equation}\label{eq:pricing-identities}
 \pi_t=\frac{e^{-r^Rt}}{(1+w_0)^\gamma}M_t^R(1+w_t)^\gamma,
 \qquad
 M_t^Rw_t^\gamma=w_0^\gamma e^{-\gamma\sigma^2H_t}M_t^N.
\end{equation}
The second identity shows how net disagreement changes the relative weight of
the dogmatist martingale component. 

Since $\eta_t$ varies with time, we use $(s,w)$ as the state. For any payoff
$F$ with finite expectation, define the time-homogeneous pricing operator
\begin{equation}\label{eq:space-time-operator}
 (\mathcal{P}_\tau F)(s,w):=
 \mathbb{E}_{s,w}^P\left[\frac{\pi_{s+\tau}}{\pi_s}
 F(s+\tau,w_{s+\tau})\right].
\end{equation}
Here $\mathbb{E}_{s,w}^P$ conditions on $w_s=w$. The operator gives the time-$s$ price
of $F(s+\tau,w_{s+\tau})$ delivered at $s+\tau$.

\begin{lemma}[Candidate Hansen--Scheinkman factorizations]
\label{lem:factorizations}
For $a,b\geq0$ with $a+b>0$ and $(s,w)\in[0,\infty)\times(0,\infty)$, define
\begin{equation}\label{eq:Phi-ab}
 \Phi_{a,b}(s,w):=
 \frac{a+b e^{\gamma\sigma^2H_s}w^\gamma}{(1+w)^\gamma},
 \qquad
 M_t^{a,b}:=\frac{aM_t^R+bw_0^\gamma M_t^N}{a+bw_0^\gamma}.
\end{equation}
Then $M^{a,b}$ is a positive martingale component with mean one. The positive
function $\Phi_{a,b}$ satisfies the eigenfunction relation
\begin{equation}\label{eq:eigen-relation}
 \mathcal{P}_\tau\Phi_{a,b}(s,w)=e^{-r^R\tau}\Phi_{a,b}(s,w).
\end{equation}
The state-price density admits the exact factorization
\begin{equation}\label{eq:exact-factorization}
 \pi_t=e^{-r^Rt}M_t^{a,b}
 \frac{\Phi_{a,b}(0,w_0)}{\Phi_{a,b}(t,w_t)}.
\end{equation}
The choices $(a,b)=(1,0)$ and $(0,1)$ give the rational and dogmatist cases,
respectively. Positive $a$ and $b$ give a convex combination of the benchmark
martingale components. In each case, $M^{a,b}$ is the candidate permanent
martingale component and the $\Phi_{a,b}$-ratio is the transitory component.
\end{lemma}

\begin{proof}
Under \eqref{eq:eta-integrability}, the martingale components $M^R$ and $M^N$
are positive exponential martingales with mean one on every finite horizon.
Hence $M^{a,b}$ is a convex combination of the two martingale components, and
$\Phi_{a,b}$ is positive. The two
identities in \eqref{eq:pricing-identities} give
\[
 \pi_t\Phi_{a,b}(t,w_t)
 =\frac{e^{-r^Rt}}{(1+w_0)^\gamma}
 \left(aM_t^R+bw_0^\gamma M_t^N\right)
 =e^{-r^Rt}\Phi_{a,b}(0,w_0)M_t^{a,b}.
\]
This proves \eqref{eq:exact-factorization}. Taking the ratio of this identity
at times $s+\tau$ and $s$ gives
\[
 \frac{\pi_{s+\tau}}{\pi_s}
 \Phi_{a,b}(s+\tau,w_{s+\tau})
 =e^{-r^R\tau}\Phi_{a,b}(s,w_s)
 \frac{M_{s+\tau}^{a,b}}{M_s^{a,b}}.
\]
From the martingale property we can conclude
\[
 \mathbb{E}^P\left[\left.\frac{M_{s+\tau}^{a,b}}{M_s^{a,b}}
 \right|\mathcal F_s\right]=1.
\]
Conditional expectation given $w_s=w$ now proves
\eqref{eq:eigen-relation}.
\end{proof}

When long-maturity bonds and forward measures converge, their limits determine
the factorization relevant to long-run pricing. For $0\leq t<T$, set $\tau:=T-t$ and
$M_{t,T}^j:=M_T^j/M_t^j$ for $j\in\{R,N\}$. We define the ratio of the dogmatist and rational valuation terms by,
\begin{equation}\label{eq:q-conditional}
 q_{t,T}:=w_t^\gamma e^{-\gamma\sigma^2H_{t,T}}.
\end{equation}

A zero-coupon bond pays one unit of consumption at $T$. Its time-$t$ price
$P_t^T$ and the conditional density $\widehat Z_{t,T}$ of the maturity-$T$
forward measure relative to $P$ are
\begin{equation}\label{eq:bond-forward-definitions}
 P_t^T:=\mathbb{E}^P\left[\frac{\pi_T}{\pi_t}\,\middle|\,\mathcal F_t\right],
 \qquad
 \widehat Z_{t,T}:=\frac{\pi_T/\pi_t}{P_t^T}.
\end{equation}
$\widehat Z_{t,T}$ is nonnegative and $\mathbb{E}^P[\widehat Z_{t,T}\mid\mathcal F_t]=1$. For any payoff $X_T$, its
conditional expectation under the maturity-$T$ forward measure is
$\mathbb{E}^P[\widehat Z_{t,T}X_T\mid\mathcal F_t]$.

To compare forward measures as $T$ increases, fix $t$ and consider each
measure on the same $\mathcal F_t$. Its density relative to
$P$ is
\begin{equation}\label{eq:ZT}
 Z_t^T:=\left.\frac{\,\mathrm{d} P^T}{\,\mathrm{d} P}\right|_{\mathcal F_t}
 =\frac{\pi_tP_t^T}{P_0^T},\qquad t\leq T.
\end{equation}
Following \citet{qin2017long}, the long-forward measure exists when the
restricted measures converge in total variation for every fixed $t$. The long
bond exists when there is a strictly positive semimartingale $B^\infty$ such
that $P_t^T/P_0^T$ converges to $B_t^\infty$ uniformly on bounded time
intervals in probability.

\begin{theorem}[Finite-maturity valuation and long-run pricing limits]
\label{thm:main}
There is a constant $K_\gamma<\infty$, depending only on $\gamma$. Set
\begin{equation}\label{eq:approximation-error}
 e_{t,T}:=K_\gamma
 \exp\left\{-\frac{\gamma}{2(1+\gamma)^2}
 \sigma^2A_{t,T}\right\}.
\end{equation}

\emph{(i) Finite-maturity approximation.} For every $0\leq t<T$,
\begin{equation}\label{eq:bond-approximation}
 P_t^T=e^{-r^R\tau}(1+w_t)^{-\gamma}(1+q_{t,T})
 (1+\varepsilon_{t,T}),
 \qquad
 |\varepsilon_{t,T}|\leq e_{t,T}.
\end{equation}
If $e_{t,T}<1$, then
\begin{equation}\label{eq:forward-approximation}
 \mathbb{E}^P\left[\left|\widehat Z_{t,T}-
 \frac{M_{t,T}^R+q_{t,T}M_{t,T}^N}{1+q_{t,T}}\right|
 \,\middle|\,\mathcal F_t\right]
 \leq\frac{2e_{t,T}}{1-e_{t,T}}.
\end{equation}
Both approximations are exact when $\gamma=1$.

\emph{(ii) Long-forward measure and long bond.} Suppose total disagreement
$A_T\to\infty$. Define the dogmatist's time-zero valuation weight for a
payment at $T$ by
\begin{equation}\label{eq:theta-T}
 \theta_T:=\frac{w_0^\gamma e^{-\gamma\sigma^2H_T}}
 {1+w_0^\gamma e^{-\gamma\sigma^2H_T}}.
\end{equation}
For every fixed $t$,
\begin{equation}\label{eq:fixed-t-approximation}
\mathbb{E}^P\left|Z_t^T-\bigl\{(1-\theta_T)M_t^R+\theta_TM_t^N\bigr\}\right|
 \longrightarrow0.
\end{equation}
The long-forward measure and long bond exist if and only if net disagreement
$H_T$ converges to a finite value or tends to $+\infty$ or $-\infty$. Let $h$ denote the
corresponding limit and set
\begin{equation}\label{eq:theta-h}
 \theta_h:=\begin{cases}
 0,&h=+\infty,\\
 \displaystyle\frac{w_0^\gamma e^{-\gamma\sigma^2h}}
 {1+w_0^\gamma e^{-\gamma\sigma^2h}},
   &h\in\mathbb R,\\
 1,&h=-\infty,
 \end{cases}
 \qquad M_t^h:=(1-\theta_h)M_t^R+\theta_hM_t^N.
\end{equation}
For every fixed $t$, $Z_t^T\to M_t^h$ in $L^1(P)$. The normalized bond-price
processes converge uniformly on bounded time intervals in probability. For
every $K<\infty$,
\begin{equation}\label{eq:long-bond}
 \sup_{0\leq s\leq K}\left|\frac{P_s^T}{P_0^T}-B_s^\infty\right|
 \xrightarrow[T\to\infty]{P}0,
 \qquad B_s^\infty:=\frac{M_s^h}{\pi_s},
 \qquad \pi_s=\frac{M_s^h}{B_s^\infty}.
\end{equation}
The permanent martingale component $M^h$ in the factorization of $\pi$ equals
$M^R$ when $H_T\to+\infty$ and $M^N$ when $H_T\to-\infty$. When $H_T$ has a
finite limit, both benchmark martingale components receive the positive weights
in \eqref{eq:theta-h}. If $H_T$ has no such limit, neither the long-forward
measure nor the long bond exists.

\emph{(iii) Long yield.} Under the same assumption $A_T\to\infty$, suppose the
time average of the belief distortion has the finite limit
$\bar\eta:=\lim_{T\to\infty}H_T/T$. Then
\begin{equation}\label{eq:long-yield}
 y^\infty:=-\lim_{T\to\infty}\frac1T\log P_0^T
 =r^R+\gamma\sigma^2\min\{\bar\eta,0\}.
\end{equation}
Here $y^\infty$ is the limiting continuously compounded yield implied by the
time-zero bond price $P_0^T$ as $T$ increases.
In particular, $H_T=o(T)$ gives $y^\infty=r^R$, even if the long-forward
measure and long bond fail to exist.
\end{theorem}

\begin{proof}
\emph{Part (i).}
The two identities in \eqref{eq:pricing-identities} imply
\begin{align*}
 P_t^T&=e^{-r^R\tau}(1+w_t)^{-\gamma}
 \mathbb{E}^{P^{L,R}}[(1+w_T)^\gamma\mid\mathcal F_t],\\
 M_{t,T}^Rw_T^\gamma&=q_{t,T}M_{t,T}^N.
\end{align*}
Under $P^{L,R}$, the process $(B_t^{L,R})_{t\geq0}$ defined by
$B_t^{L,R}:=B_t+\gamma\sigma t$ is a Brownian motion. Evaluating the moment of $w_T$ gives, for $u\in(0,\gamma)$,
\begin{equation}\label{eq:conditional-moment}
 \mathbb{E}^{P^{L,R}}[w_T^u\mid\mathcal F_t]
 =q_{t,T}^{u/\gamma}
 \exp\left\{-\frac12\frac{u}{\gamma}
 \left(1-\frac{u}{\gamma}\right)\sigma^2A_{t,T}\right\}.
\end{equation}

Let $R_\gamma(x):=(1+x)^\gamma-1-x^\gamma$. The ratio
\[
 \frac{|R_\gamma(x)|}
 {x^{\gamma/(1+\gamma)}+x^{\gamma^2/(1+\gamma)}}
\]
is continuous on $(0,\infty)$ and is bounded as $x$ approaches 0 or
$\infty$. Hence a finite constant $C_\gamma$ satisfies
\begin{equation}\label{eq:R-bound}
 |R_\gamma(x)|\leq C_\gamma
 \left(x^{\gamma/(1+\gamma)}+x^{\gamma^2/(1+\gamma)}\right),
 \qquad x>0.
\end{equation}
Applying \eqref{eq:conditional-moment} to the two powers in
\eqref{eq:R-bound} and using $x^a\leq1+x$ for $a\in(0,1)$ gives
\[
 \mathbb{E}^{P^{L,R}}[|R_\gamma(w_T)|\mid\mathcal F_t]
 \leq 2C_\gamma(1+q_{t,T})
 \exp\left\{-\frac{\gamma}{2(1+\gamma)^2}
 \sigma^2A_{t,T}\right\}.
\]
We may therefore take $K_\gamma=2C_\gamma$.

Set
\begin{equation*}
 X_{t,T}:=M_{t,T}^R(1+w_T)^\gamma,
 \qquad
 Y_{t,T}:=M_{t,T}^R+q_{t,T}M_{t,T}^N.
\end{equation*}
The preceding identities and the change of measure defined by $M^R$ yield
\begin{equation*}
 \mathbb{E}^P[Y_{t,T}\mid\mathcal F_t]=1+q_{t,T},
 \qquad
 \mathbb{E}^P[|X_{t,T}-Y_{t,T}|\mid\mathcal F_t]
 \leq e_{t,T}(1+q_{t,T}).
\end{equation*}
Define
\[
 \varepsilon_{t,T}:=
 \frac{\mathbb{E}^P[X_{t,T}\mid\mathcal F_t]
 -\mathbb{E}^P[Y_{t,T}\mid\mathcal F_t]}
 {\mathbb{E}^P[Y_{t,T}\mid\mathcal F_t]}.
\]
The preceding bound gives $|\varepsilon_{t,T}|\leq e_{t,T}$ and proves
\eqref{eq:bond-approximation}. Moreover,
\[
 \widehat Z_{t,T}
 =\frac{X_{t,T}}{\mathbb{E}^P[X_{t,T}\mid\mathcal F_t]},
 \qquad
 \frac{M_{t,T}^R+q_{t,T}M_{t,T}^N}{1+q_{t,T}}
 =\frac{Y_{t,T}}{\mathbb{E}^P[Y_{t,T}\mid\mathcal F_t]}.
\]
If $e_{t,T}<1$, then
$\mathbb{E}^P[X_{t,T}\mid\mathcal F_t]\geq
(1-e_{t,T})\mathbb{E}^P[Y_{t,T}\mid\mathcal F_t]$, and hence
\begin{align*}
 \mathbb{E}^P\left[\left|
 \frac{X_{t,T}}{\mathbb{E}^P[X_{t,T}\mid\mathcal F_t]}-
 \frac{Y_{t,T}}{\mathbb{E}^P[Y_{t,T}\mid\mathcal F_t]}
 \right|\,\middle|\,\mathcal F_t\right]
 &\leq
 \frac{2\mathbb{E}^P[|X_{t,T}-Y_{t,T}|\mid\mathcal F_t]}
 {\mathbb{E}^P[X_{t,T}\mid\mathcal F_t]}\\
 &\leq\frac{2e_{t,T}}{1-e_{t,T}}.
\end{align*}
This proves \eqref{eq:forward-approximation}. When $\gamma=1$,
$R_\gamma$ vanishes, so both approximations are exact.

\emph{Part (ii).}
The assumption $A_T\to\infty$ implies $e_{0,T}\to0$. For
$a\in[0,1]$, define
\[
 \overline{M}_t(a):=(1-a)M_t^R+aM_t^N.
\]
At time zero, the difference in \eqref{eq:forward-approximation} is
$\widehat Z_{0,T}-\overline{M}_T(\theta_T)$. Since $M^R$ and $M^N$ are martingales,
\[
 Z_t^T-\overline{M}_t(\theta_T)
 =\mathbb{E}^P[\widehat Z_{0,T}-\overline{M}_T(\theta_T)\mid\mathcal F_t].
\]
Jensen's inequality and the time-zero bound in
\eqref{eq:forward-approximation} therefore give, for every fixed $t$ and all
sufficiently large $T$,
\begin{equation}\label{eq:fixed-t-rate}
 \mathbb{E}^P|Z_t^T-\overline{M}_t(\theta_T)|
 \leq\frac{2e_{0,T}}{1-e_{0,T}}.
\end{equation}
This proves \eqref{eq:fixed-t-approximation}.

If $H_T$ has one of the three limits in the theorem, then
$\theta_T\to\theta_h$. Equation \eqref{eq:fixed-t-rate} gives
$Z_t^T\to M_t^h$ in $L^1(P)$ for every fixed $t$. Hence the forward measures
converge in total variation on every fixed $\mathcal F_t$.

We next show that $Z^T=(Z_s^T)_{0\leq s\leq T}$ converges to $M^h$ uniformly
on bounded time intervals in probability. Set
$L_s^T:=Z_s^T-\overline M_s(\theta_T)$ for $0\leq s\leq T$. The conditional
expectation identity above shows that $L^T$ is a martingale. For $T>K$, Doob's
maximal inequality gives, for every $\epsilon>0$,
\[
 P\left(\sup_{0\leq s\leq K}|L_s^T|>\epsilon\right)
 \leq\frac{\mathbb{E}^P|L_K^T|}{\epsilon}
 \leq\frac{2e_{0,T}}{\epsilon(1-e_{0,T})}
 \longrightarrow0.
\]
Moreover,
\[
 \sup_{0\leq s\leq K}
 |\overline M_s(\theta_T)-M_s^h|
 =|\theta_T-\theta_h|\sup_{0\leq s\leq K}|M_s^N-M_s^R|
 \longrightarrow0
 \qquad P\text{-a.s.}
\]
Thus $Z^T\to M^h$ uniformly on bounded time intervals in probability. Since
$\pi$ is continuous and strictly positive,
\[
 \sup_{0\leq s\leq K}\left|\frac{P_s^T}{P_0^T}-\frac{M_s^h}{\pi_s}\right|
 \leq \sup_{0\leq s\leq K}\frac1{\pi_s}
       \sup_{0\leq s\leq K}|Z_s^T-M_s^h|
 \xrightarrow[T\to\infty]{P}0.
\]
This proves \eqref{eq:long-bond}.
The corresponding factorization follows from
Lemma \ref{lem:factorizations}. Its
coefficients are $(1,e^{-\gamma\sigma^2h})$ for finite $h$, $(1,0)$ when
$H_T\to+\infty$, and $(0,1)$ when $H_T\to-\infty$.

We next show that convergence of the forward measures requires $\theta_T$ to
converge. Since $A_T\to\infty$, there is a time $t_0$ with $A_{t_0}>0$. Therefore, $\log(M_{t_0}^N/M_{t_0}^R)$ is normally distributed with
variance $\sigma^2A_{t_0}>0$. Hence $M_{t_0}^R$ and $M_{t_0}^N$ differ with positive
probability. Consequently,
\begin{equation}\label{eq:weight-distance}
 \mathbb{E}^P|\overline{M}_{t_0}(a)-\overline{M}_{t_0}(b)|
 =d_{t_0}|a-b|,
 \qquad
 d_{t_0}:=\mathbb{E}^P|M_{t_0}^N-M_{t_0}^R|>0.
\end{equation}
If the forward measures converge in total variation on $\mathcal F_{t_0}$,
then $Z_{t_0}^T$ converges in $L^1(P)$. If the normalized bond-price processes
converge uniformly on bounded time intervals in probability, then
$P_{t_0}^T/P_0^T$ converges in probability, and so does
$Z_{t_0}^T=\pi_{t_0}P_{t_0}^T/P_0^T$. In either case,
\eqref{eq:fixed-t-rate} implies that
$\overline M_{t_0}(\theta_T)$ converges in probability to the same limit.

If $\theta_T$ failed to converge, there would be sequences $T_n,S_n\to\infty$
along which $\theta_{T_n}\to a$ and $\theta_{S_n}\to b$ for distinct
$a,b\in[0,1]$. The corresponding weighted combinations would converge almost
surely to $\overline M_{t_0}(a)$ and $\overline M_{t_0}(b)$. Both random
variables would equal the common probability limit.
Equation \eqref{eq:weight-distance} would
then imply $a=b$, a contradiction. Hence $\theta_T$ converges.
Equation \eqref{eq:theta-T} shows that this occurs exactly when $H_T$
converges to a finite value or tends to $+\infty$ or $-\infty$.

\emph{Part (iii).}
At time zero, \eqref{eq:bond-approximation} gives
\begin{equation}\label{eq:yield-expansion}
 -\frac1T\log P_0^T
 =r^R+\frac{\gamma\log(1+w_0)-
 \log(1+w_0^\gamma e^{-\gamma\sigma^2H_T})}{T}+o(1).
\end{equation}
The error term vanishes because $e_{0,T}\to0$. If
$H_T/T\to\bar\eta$, then
\[
 \frac1T\log(1+w_0^\gamma e^{-\gamma\sigma^2H_T})
 \longrightarrow\max\{-\gamma\sigma^2\bar\eta,0\}.
\]
Substitution in \eqref{eq:yield-expansion} proves
\eqref{eq:long-yield}.
\end{proof}

At finite maturities, the approximation error falls as total disagreement
$A_{t,T}$ grows. The relative weights of the benchmark valuation laws depend on
net disagreement $H_{t,T}$. The long yield depend only on the average $H_T/T$.
This difference explains why repeated changes in net disagreement can keep
valuation weights dependent on maturity after the long yield has converged.
Where a dogmatist investor has a constant forecast error from the rational investor, $\eta_t:=\eta$, $H_T=\eta T$, the sign of the $\eta$, the direction of the forecast error, determines the valuation law associated with long-run pricing. 

\section{Finite-maturity pricing under time-varying beliefs}\label{sec:finite}

We now study how belief disagreement affects finite-maturity bond prices. For
tractability, we set $\gamma=1$ and assume that $\eta_t$ is continuous. Under this
preference specification, the two approximations in
Theorem \ref{thm:main}(i) are exact. The cumulative forecast difference over $[t,T]$ is $G_{t,T}$.
\begin{equation}\label{eq:log-definitions}
 \delta_t:=\sigma^2\eta_t,
 \qquad G_{t,T}:=\int_t^T\delta_s\,\mathrm{d} s.
\end{equation}
The dogmatist's current consumption share, the valuation odds in
\eqref{eq:q-conditional}, and the maturity-$T$ valuation weight are
\begin{equation}\label{eq:log-weights}
 s_t:=\frac{w_t}{1+w_t},
 \qquad
 q_{t,T}=w_te^{-G_{t,T}},
 \qquad \theta_{t,T}:=\frac{q_{t,T}}{1+q_{t,T}}.
\end{equation}

For $T>t$, let $y_t^T:=-(T-t)^{-1}\log P_t^T$ denote the yield and let
$f_t(T):=-\partial_T\log P_t^T$ denote the instantaneous forward rate. The diffusion term in the bond return
$\,\mathrm{d} P_t^T/P_t^T$ is $\beta_t^T\,\mathrm{d} B_t$. Thus $|\beta_t^T|$ is the
instantaneous bond volatility, and its sign gives the direction of the bond's
response to an endowment shock. The coefficient $\mu_t^T$ is the expected
growth rate of aggregate endowment under the maturity-$T$ forward measure.

\begin{proposition}
\label{prop:log-formulas}
Suppose $\gamma=1$, so $r^R=\rho+\mu-\sigma^2$. For every $0\leq t<T$,
\begin{align}
 P_t^T&=e^{-r^R(T-t)}\frac{1+q_{t,T}}{1+w_t},
 \label{eq:log-bond}\\
 \widehat Z_{t,T}&=(1-\theta_{t,T})M_{t,T}^R
 +\theta_{t,T}M_{t,T}^N,
 \label{eq:log-density}\\
 y_t^T-r^R&=\frac{\log(1+w_t)-\log(1+q_{t,T})}{T-t},
 \label{eq:log-yield}\\
 f_t(T)-r^R&=\delta_T\theta_{t,T},
 \label{eq:log-forward-rate}\\
 \beta_t^T&=\eta_t\sigma(\theta_{t,T}-s_t),
 \label{eq:log-bond-volatility}\\
 \mu_t^T&=\mu-\sigma^2+\delta_t\theta_{t,T}.
 \label{eq:log-forward-drift}
\end{align}
Under the maturity-$T$ forward law, endowment satisfies
$\,\mathrm{d} D_t/D_t=\mu_t^T\,\mathrm{d} t+\sigma\,\mathrm{d} B_t^T$.
Here $(B_t^T)_{t\geq0}$ is a Brownian motion under that law.
\end{proposition}

\begin{proof}
For $\gamma=1$, the remainder in Theorem \ref{thm:main}(i) vanishes, which
gives \eqref{eq:log-bond} and \eqref{eq:log-density}. Taking logarithms of
\eqref{eq:log-bond} gives \eqref{eq:log-yield}. Differentiation with respect
to $T$ gives \eqref{eq:log-forward-rate}. The diffusion coefficients of $w_t$
and $q_{t,T}$ are both $\eta_t\sigma$. It\^o's formula applied to
\eqref{eq:log-bond} gives \eqref{eq:log-bond-volatility}. Under $P$, $Z^T$ has
diffusion coefficient $-\sigma+\eta_t\sigma\theta_{t,T}$. Applying Girsanov's theorem
then gives \eqref{eq:log-forward-drift}.
\end{proof}

Proposition \ref{prop:log-formulas} shows that the same valuation weight
$\theta_{t,T}$ enters the yield, the forward rate at $T$, the bond's exposure
to current endowment shocks, and the forward-law growth rate.

Set $s_0=5\%$ and let pessimism fade according to
\begin{equation}\label{eq:numerical-path}
 \eta_t=\frac{\eta_0}{\sqrt{1+t}},
 \qquad
 H_T=2\eta_0(\sqrt{1+T}-1),
 \qquad
 G_{0,T}=\sigma^2H_T,
 \qquad \eta_0<0.
\end{equation}
Here $\delta_t=\sigma^2\eta_t$ is the difference between dogmatist's expected-growth forecast and the rational investor's forecast. The five cases are ordered by
$|\delta_0|$. Along each path, $A_T\to\infty$, $H_T\to-\infty$, and
$H_T/T\to0$. The dogmatist therefore disappears from consumption, her
benchmark valuation law describes the long-forward limit, and the long yield
equals $r^R$.

Panels (a) and (b) of Figure \ref{fig:numerical} show how fading pessimism
affects the dogmatist's valuation weight and the yield spread. The valuation
weight moves above the dogmatist's time-zero consumption share and increases
with maturity. A more negative $\eta_0$ and a higher $\sigma$ strengthen this
response through $\sigma^2H_T$. Milder pessimism and lower volatility keep the
valuation weight closer to $s_0$. The yield spread is negative since pessimism
raises the value of future consumption. The cumulative forecast difference is
of order $\sqrt T$. Its effect on the annualized yield vanishes as
$T\to\infty$, so the yield spread converges to zero.

Table \ref{tab:finite-comparisons} reports the bond-price premium
$\Delta P_{0,T}:=e^{r^RT}P_0^T-1$ and the time-zero value of the dogmatist's
consumption after $T$, denoted by
$\mathcal W_N(T):=\mathbb{E}^P[\int_T^\infty\pi_tc_{N,t}\,\mathrm{d} t]$. Under $\gamma=1$,
$\mathcal W_N(T)/(W_R+W_N)=s_0e^{-\rho T}$. We use $\rho=2\%$ for this
calculation.

Alternating paths provide a reduced-form representation of procyclical growth
expectations that adjust slowly at business-cycle turning points.
\citet{dovern2017systematic} find that professional growth forecasts are
systematically too optimistic for recessions and too pessimistic for
recoveries. To study amplitude decay, set $\sigma=15\%$, $L=20$ years, and
$\omega=2\pi/L$. Consider
\begin{equation}\label{eq:alternative-paths}
 \eta_t^{\mathrm{FA}}=\frac{2\sin(\omega t)}{\sqrt{1+t}},
 \qquad
 \eta_t^{\mathrm{CA}}=2\cos(\omega t).
\end{equation}
Both paths have $A_T\to\infty$, so the dogmatist disappears from consumption.
Under fading alternation, $G_{0,T}^{\mathrm{FA}}$ converges. Her valuation
weight approaches a value strictly between zero and one, and the long-forward
measure assigns positive weight to each benchmark valuation law.

Under constant-amplitude cycles, the cumulative forecast difference is
\begin{equation}\label{eq:cyclical-cumulative-forecast-difference}
 G_{0,T}^{\mathrm{CA}}=\frac{2\sigma^2}{\omega}\sin(\omega T).
\end{equation}
This quantity is bounded and has no limit. The valuation weight inherits
this maturity dependence, so the long-forward measure and long bond fail to
converge. Since $G_{0,T}^{\mathrm{CA}}/T\to0$, the long yield equals $r^R$.

Panels (c) and (d) of Figure \ref{fig:numerical} trace the finite-maturity
implications. The fading path approaches a stable valuation weight.
Constant-amplitude cycles keep the valuation weight oscillating as the yield
spread approaches zero.

\begin{figure}[!ht]
\centering
\includegraphics[width=0.95\textwidth]{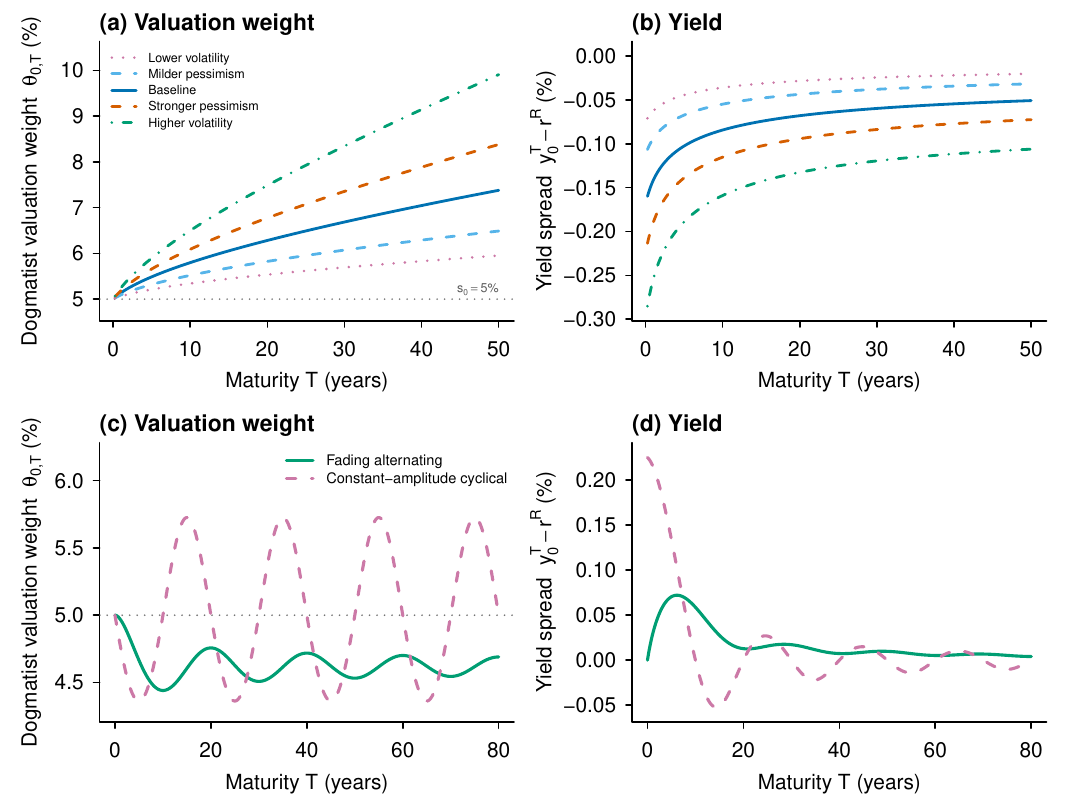}
\caption{Finite-maturity pricing under fading and alternating disagreement.
Panel (a) varies $\eta_0$ with $\sigma=15\%$. Panel (b) varies
$\sigma$ with $\eta_0=-1.5$. Panels (c) and (d) show the paths in
\eqref{eq:alternative-paths}, with $\sigma=15\%$ and $L=20$ years. The
parameter values for the named cases appear in Table \ref{tab:finite-comparisons}.
The time-zero consumption share is $s_0=5\%$. The dotted line in each
valuation-weight panel marks this benchmark.}
\label{fig:numerical}
\end{figure}

\begin{table}[!ht]
\centering
\caption{Finite-maturity effects of fading pessimism. Panel A compares the five
cases. Panel B reports additional quantities for the baseline. Valuation and
pricing quantities are in percentage units. We use $s_0=5\%$ and $\rho=2\%$.}
\label{tab:finite-comparisons}
\footnotesize
\textit{Panel A. Valuation weights and yields}\par
\setlength{\tabcolsep}{3.5pt}
\begin{tabular}{lccrrrrrr}
\toprule
&&& \multicolumn{3}{c}{Valuation weight $\theta_{0,T}$} &
\multicolumn{3}{c}{Yield spread $y_0^T-r^R$} \\
\cmidrule(lr){4-6}\cmidrule(lr){7-9}
Case & $\eta_0$ & $\sigma$ & $T=10$ & $T=20$ & $T=30$ &
$T=10$ & $T=20$ & $T=30$ \\
\midrule
Lower volatility    & $-1.50$ & $10\%$ & 5.34 & 5.54 & 5.69 & $-0.036$ & $-0.028$ & $-0.024$ \\
Milder pessimism    & $-1.00$ & $15\%$ & 5.52 & 5.82 & 6.07 & $-0.055$ & $-0.044$ & $-0.038$ \\
Baseline            & $-1.50$ & $15\%$ & 5.80 & 6.28 & 6.69 & $-0.084$ & $-0.068$ & $-0.060$ \\
Stronger pessimism  & $-2.00$ & $15\%$ & 6.09 & 6.77 & 7.36 & $-0.115$ & $-0.094$ & $-0.084$ \\
Higher volatility   & $-1.50$ & $20\%$ & 6.50 & 7.48 & 8.35 & $-0.159$ & $-0.133$ & $-0.120$ \\
\bottomrule
\end{tabular}

\smallskip
\textit{Panel B. Baseline economic magnitudes}\par
\setlength{\tabcolsep}{5pt}
\begin{tabular}{rrrrrr}
\toprule
$T$ & \shortstack{Bond-price\\premium\\$\Delta P_{0,T}$}
& \shortstack{Forward-rate\\spread\\$f_0(T)-r^R$}
& \shortstack{Bond\\volatility\\$|\beta_0^T|$}
& \shortstack{Forward-law\\growth shift\\$\mu_0^T-(\mu-\sigma^2)$}
& \shortstack{Future-consumption\\value\\$\mathcal W_N(T)/(W_R+W_N)$} \\
\midrule
10 & 0.846 & $-0.059$ & 0.179 & $-0.196$ & 4.094 \\
20 & 1.368 & $-0.046$ & 0.288 & $-0.212$ & 3.352 \\
30 & 1.806 & $-0.041$ & 0.379 & $-0.226$ & 2.744 \\
\bottomrule
\end{tabular}
\end{table}

\FloatBarrier

Panel B shows that narrow yield spreads can coexist with rising bond-price
premiums, measurable bond exposures, and future consumption worth several
percent of initial aggregate wealth.
\section{Conclusion}\label{sec:scope}

Physical extinction does not identify whose belief shapes long-run pricing.
Total disagreement can drive the dogmatist's consumption share to zero. Net
disagreement determines her contribution to distant bond valuation because
optimistic and pessimistic forecast errors can offset one another.

A finite limit of net disagreement gives a stable valuation law that reflects
both investors. Constant-amplitude cycles prevent the long-forward measure and
long bond from existing even though the long yield converges. At finite
maturities, valuation weights and bond exposures can change when yield
differences are small. The long yield therefore cannot reveal whether the
valuation law and normalized bond prices converge.

\noindent\textbf{Use of artificial intelligence.} The author used ChatGPT to
assist with restructuring the article and language polishing. The author
assumes responsibility for all content.

\begingroup
\footnotesize
\setlength{\bibsep}{0pt}
\bibliographystyle{abbrvnat}
\bibliography{Paper}

@article{blume2006smart,
  author  = {Blume, Lawrence and Easley, David},
  title   = {If You're So Smart, Why Aren't You Rich? Belief Selection in Complete and Incomplete Markets},
  journal = {Econometrica},
  year    = {2006},
  volume  = {74},
  number  = {4},
  pages   = {929--966},
  doi     = {10.1111/j.1468-0262.2006.00691.x}
}

@article{borovicka2020survival,
  author  = {Borovi{\v{c}}ka, Jaroslav},
  title   = {Survival and Long-Run Dynamics with Heterogeneous Beliefs under Recursive Preferences},
  journal = {Journal of Political Economy},
  year    = {2020},
  volume  = {128},
  number  = {1},
  pages   = {206--251},
  doi     = {10.1086/704072}
}

@article{borovicka2016misspecified,
  author  = {Borovi{\v{c}}ka, Jaroslav and Hansen, Lars Peter and Scheinkman, Jos{\'e} A.},
  title   = {Misspecified Recovery},
  journal = {The Journal of Finance},
  year    = {2016},
  volume  = {71},
  number  = {6},
  pages   = {2493--2544},
  doi     = {10.1111/jofi.12404}
}

@article{cvitanic2011price,
  author  = {Cvitani{\'c}, Jak{\v{s}}a and Malamud, Semyon},
  title   = {Price Impact and Portfolio Impact},
  journal = {Journal of Financial Economics},
  year    = {2011},
  volume  = {100},
  number  = {1},
  pages   = {201--225},
  doi     = {10.1016/j.jfineco.2010.11.001}
}

@article{cvitanic2012financial,
  author  = {Cvitani{\'c}, Jak{\v{s}}a and Jouini, Ely{\`e}s and Malamud, Semyon and Napp, Clotilde},
  title   = {Financial Markets Equilibrium with Heterogeneous Agents},
  journal = {Review of Finance},
  year    = {2012},
  volume  = {16},
  number  = {1},
  pages   = {285--321},
  doi     = {10.1093/rof/rfr018}
}

@article{dovern2017systematic,
  author  = {Dovern, Jonas and Jannsen, Nils},
  title   = {Systematic Errors in Growth Expectations over the Business Cycle},
  journal = {International Journal of Forecasting},
  year    = {2017},
  volume  = {33},
  number  = {4},
  pages   = {760--769},
  doi     = {10.1016/j.ijforecast.2017.03.003}
}

@article{hansen2009long,
  author  = {Hansen, Lars Peter and Scheinkman, Jos{\'e} A.},
  title   = {Long-Term Risk: An Operator Approach},
  journal = {Econometrica},
  year    = {2009},
  volume  = {77},
  number  = {1},
  pages   = {177--234},
  doi     = {10.3982/ECTA6761}
}

@article{kogan2006price,
  author  = {Kogan, Leonid and Ross, Stephen A. and Wang, Jiang and Westerfield, Mark M.},
  title   = {The Price Impact and Survival of Irrational Traders},
  journal = {The Journal of Finance},
  year    = {2006},
  volume  = {61},
  number  = {1},
  pages   = {195--229},
  doi     = {10.1111/j.1540-6261.2006.00834.x}
}

@article{malamud2008long,
  author  = {Malamud, Semyon},
  title   = {Long Run Forward Rates and Long Yields of Bonds and Options in Heterogeneous Equilibria},
  journal = {Finance and Stochastics},
  year    = {2008},
  volume  = {12},
  number  = {2},
  pages   = {245--264},
  doi     = {10.1007/s00780-007-0058-0}
}

@article{qin2016positive,
  author  = {Qin, Likuan and Linetsky, Vadim},
  title   = {Positive Eigenfunctions of {Markovian} Pricing Operators: {Hansen--Scheinkman} Factorization, {Ross} Recovery, and Long-Term Pricing},
  journal = {Operations Research},
  year    = {2016},
  volume  = {64},
  number  = {1},
  pages   = {99--117},
  doi     = {10.1287/opre.2015.1449}
}

@article{qin2017long,
  author  = {Qin, Likuan and Linetsky, Vadim},
  title   = {Long-Term Risk: A Martingale Approach},
  journal = {Econometrica},
  year    = {2017},
  volume  = {85},
  number  = {1},
  pages   = {299--312},
  doi     = {10.3982/ECTA13438}
}

@article{qin2018hjm,
  author  = {Qin, Likuan and Linetsky, Vadim},
  title   = {Long-Term Factorization in {Heath--Jarrow--Morton} Models},
  journal = {Finance and Stochastics},
  year    = {2018},
  volume  = {22},
  number  = {3},
  pages   = {621--641},
  doi     = {10.1007/s00780-018-0365-7}
}

@article{sandroni2000markets,
  author  = {Sandroni, Alvaro},
  title   = {Do Markets Favor Agents Able to Make Accurate Predictions?},
  journal = {Econometrica},
  year    = {2000},
  volume  = {68},
  number  = {6},
  pages   = {1303--1341},
  doi     = {10.1111/1468-0262.00163}
}

@article{yan2008natural,
  author  = {Yan, Hongjun},
  title   = {Natural Selection in Financial Markets: Does It Work?},
  journal = {Management Science},
  year    = {2008},
  volume  = {54},
  number  = {11},
  pages   = {1935--1950},
  doi     = {10.1287/mnsc.1080.0911}
}
\endgroup

\end{document}